\documentclass[a4paper]{amsart}

\title[Option pricing in bilateral Gamma stock models]{Option pricing in bilateral Gamma stock models}
\author{Uwe K{\"u}chler \and Stefan Tappe}
\address{Humboldt University of Berlin, Institute of Mathematics, Unter den Linden 6, D-10099 Berlin, Germany; ETH Z\"urich, Department of Mathematics, R\"amistrasse 101, CH-8092 Z\"urich, Switzerland}
\email{kuechler@mathematik.hu-berlin.de, stefan.tappe@math.ethz.ch}

\usepackage{amscd}
\usepackage{amsmath}
\usepackage{amssymb}
\usepackage{amsthm}
\usepackage{bbm}

\newif\ifpdf
\ifx\pdfoutput\undefined
   \pdffalse        
\else
   \pdfoutput=1     
   \pdftrue
\fi

\ifpdf
   \usepackage[pdftex]{graphicx}
\else
   \usepackage{graphicx}
\fi

\numberwithin{equation}{section}
\swapnumbers

\newtheorem{satz}{Satz}[section]

\newtheorem{theorem}[satz]{Theorem}
\newtheorem{proposition}[satz]{Proposition}
\newtheorem{corollary}[satz]{Corollary}
\newtheorem{lemma}[satz]{Lemma}

\newtheorem{definition}[satz]{Definition}

\newtheorem{remark}[satz]{Remark}

\begin{document}

\maketitle\thispagestyle{empty}

\begin{abstract}
In the framework of bilateral Gamma stock models we seek for
adequate option pricing measures, which have an economic
interpretation and allow numerical calculations of option prices.
Our investigations encompass Esscher transforms, minimal entropy
martingale measures, $p$-optimal martingale measures, bilateral Esscher transforms and the minimal martingale measure. We illustrate
our theory by a numerical example.

\bigskip

\end{abstract}


\section{Introduction}

An issue in continuous time finance is to find realistic and
analytically tractable models for price evolutions of risky
financial assets. In this text, we consider exponential L\'evy
models
\begin{align}\label{exp-Levy-model}
\left\{
\begin{array}{rcl}
S_t & = & S_0 e^{X_t}
\\ B_t & = & e^{r t}
\end{array}
\right.
\end{align}
consisting of two financial assets $(S,B)$, one dividend paying
stock $S$ with dividend rate $q \geq 0$, where $X$ denotes a L\'evy
process and one risk free asset $B$, the bank account with fixed
interest rate $r \geq 0$. Note that the classical Black-Scholes
model is a special case by choosing $X_t = \sigma W_t + (\mu -
\frac{\sigma^2}{2})t$, where $W$ is a Wiener process, $\mu \in
\mathbb{R}$ denotes the drift and $\sigma > 0$ the volatility.

Although the Black-Scholes model is a standard model in financial
mathematics, it is well-known that it only provides a poor fit to
observed market data, because typically the empirical densities
possess heavier tails and much higher located modes than fitted normal
distributions.

Several authors therefore suggested more sophisticated L\'evy
processes with a jump component. We mention the Variance Gamma
processes \cite{Madan-1990,Madan}, hyperbolic processes
\cite{Eberlein-Keller}, normal inverse Gaussian processes
\cite{Barndorff}, generalized hyperbolic processes
\cite{Eberlein-Keller-Prause,Eberlein-Prause,Prause}, CGMY processes
\cite{CGMY} and Meixner processes \cite{Schoutens-Meixner}. A survey
about L\'evy processes used for applications to finance can for
instance be found in \cite[Chap. 4]{Cont-Tankov} or \cite[Chap.
5.3]{Schoutens}.

Recently, the class of bilateral Gamma processes, which we shall
henceforth deal with in this text, was proposed in
\cite{Kuechler-Tappe}. We also mention the related article
\cite{Kuechler-Tappe-shapes}, where the shapes of their densities
are investigated.

Now, let $(S,B)$ be an exponential L\'evy model of the type
(\ref{exp-Levy-model}). In practice, we often have to deal with
adequate pricing of European options $\Phi(S_T)$, where $T > 0$
denotes the time of maturity and $\Phi : \mathbb{R} \rightarrow
\mathbb{R}$ the payoff function. For example, the payoff profile of
a European call option is $\Phi(x) = (x-K)^+$.

The option price is given by $e^{-r T}
\mathbb{E}_{\mathbb{Q}}[\Phi(S_T)]$, where $\mathbb{Q} \sim
\mathbb{P}$ is a \textit{local martingale measure}, i.e., a
probability measure, which is equivalent to the objective
probability measure $\mathbb{P}$, such that the discounted stock
price process
\begin{align}\label{discounted-price}
\tilde{S}_t := e^{-(r-q)t} S_t = S_0 e^{X_t - (r-q)t}, \quad t \geq
0
\end{align}
is a local $\mathbb{Q}$-martingale, where $q$ denotes the dividend rate. Note that $(\tilde{S}_t)_{t \geq 0}$
also has the interpretation as the discounted value process of an investor who is endowed with one stock
and reinvests all dividend payments -- $q S_t dt$ per share -- in new shares of the same stock. Indeed, integration
by parts shows that
\begin{align*}
e^{-rt} S_t = \tilde{S}_t e^{-qt} &= S_0 + \int_0^t e^{-qs}d \tilde{S}_s - \int_0^t \tilde{S}_s q e^{-qs} ds
\\ &= S_0 + \int_0^t e^{-qs}d \tilde{S}_s - \int_0^t e^{-rs} q S_s ds , \quad t \geq 0
\end{align*}
and therefore the discounted price process $(\tilde{S}_t)_{t \geq 0}$ is a local $\mathbb{Q}$-martingale if and only if the process
\begin{align*}
e^{-rt} S_t + \int_0^t e^{-rs} q S_s ds, \quad t \geq 0
\end{align*}
is a local $\mathbb{Q}$-martingale.

Typically, exponential L\'evy
models (\ref{exp-Levy-model}) are free of arbitrage, but, unlike the
classical Black-Scholes model, not complete, that is, there exist
several martingale measures $\mathbb{Q} \sim \mathbb{P}$. In
particular, if not enough market data for calibration are available,
we are therefore faced with problem, which martingale measure we
should choose.

This text is devoted to the existence of suitable pricing measures
$\mathbb{Q} \sim \mathbb{P}$ for bilateral Gamma stock models. We
will in particular focus on the following two criteria:

\begin{itemize}
\item Under $\mathbb{Q}$, the process $X$ should again be a bilateral
Gamma process, because this allows,  by virtue of the simple
characteristic function (\ref{cf-bil-Gamma}) below, numerical
calculations of option prices by using the method of Fourier
transformation.

\item The measure $\mathbb{Q}$ should have a reasonable economic
interpretation.
\end{itemize}

The basic tool for all subsequent calculations in this paper is the cumulant
generating function (\ref{cumulant}) below of the bilateral Gamma
distribution. There exists a more general class of infinitely
divisible distributions, the so-called tempered stable distributions
(see \cite[Sec. 4.5]{Cont-Tankov}), which also include the CGMY distributions. They have a comparatively simple
cumulant generating function, too, and thus analogous results for essential parts of the
present paper can also be proven for them. However, their cumulant generating functions
considerably differ from that of bilateral Gamma distributions, in particular
they have finite values at the boundary of their domains, which is not true for
cumulant generating functions of bilateral Gamma distributions.

In order to avoid a cumbersome presentation of this article,
we shall provide the concrete results for tempered stable processes
in a subsequent paper, in which we also investigate further properties of this family of distributions.
Nevertheless, in this text, we shall always indicate the corresponding results for
tempered stable distributions -- without going into detail -- and compare it with the results derived for bilateral Gamma distributions.

One practical advantage of bilateral Gamma distributions is the explicit form of the density, which we do not
have for general tempered stable distributions. This allows to
apply maximum likelihood estimators in order to determine the
parameters from observations of a stock.

The remainder of this text is organized as follows. In Section \ref{sec-stoch-calculus} we
provide the required results from stochastic calculus. In Section
\ref{sec-bilateral} we review bilateral Gamma processes and
introduce bilateral Gamma stock models. Sections \ref{sec-Esscher}--\ref{sec-mmm} are devoted
to several approaches on choosing suitable martingale measures for
option pricing in bilateral Gamma stock models. We conclude with a
numerical illustration in Section \ref{sec-numerics}.

\section{Prerequisites from stochastic
calculus}\label{sec-stoch-calculus}

In this section, we collect the results from stochastic calculus,
which we will require in the sequel. Throughout this text, let
$(\Omega,\mathcal{F},(\mathcal{F}_t)_{t \geq 0},\mathbb{P})$ be a
filtered probability space satisfying the usual conditions.

\begin{lemma}
\cite[Thm. I.4.61]{JS} Let $X$ be a real-valued semimartingale.
There exists a unique (up to indistinguishability) solution $Z$ for
the equation
\begin{align}\label{eqn-stoch-exp}
Z_t = 1 + \int_0^t Z_{s-} dX_s, \quad t \geq 0.
\end{align}
\end{lemma}

\begin{definition}
Let $X$ be a real-valued semimartingale. We call the unique solution
$Z$ for (\ref{eqn-stoch-exp}) the {\rm stochastic exponential} or
{\rm Dol\'eans-Dade exponential of $X$} and write $\mathcal{E}(X) :=
Z$.
\end{definition}

\begin{lemma}\label{lemma-stoch-log}
\cite[Lemma 2.2]{Kallsen-Shiryaev} Let $Z$ be a semimartingale such
that $Z,Z_-$ are $\mathbb{R} \setminus \{ 0 \}$-valued. There
exists a unique (up to indistinguishability) semimartingale $X$ such
that $X_0 = 0$ and $Z = Z_0 \mathcal{E}(X)$. It is given by
\begin{align}\label{eqn-stoch-log}
X_t = \int_0^t \frac{1}{Z_{s-}} dZ_s, \quad t \geq 0.
\end{align}
\end{lemma}

\begin{definition}
Let $Z$ be a semimartingale such that $Z,Z_-$ are $\mathbb{R}
\setminus \{ 0 \}$-valued. We call the unique process $X$ from Lemma
\ref{lemma-stoch-log} the {\rm stochastic logarithm of $Z$} and
write $\mathcal{L}(Z) := X$.
\end{definition}

\begin{lemma}\label{lemma-Levy-mt}
Let $X$ be a real-valued L\'evy process. The following statements
are equivalent:
\begin{enumerate}
\item $X$ is a local martingale;

\item $X$ is a martingale;

\item $\mathbb{E}[X_1] = 0$.
\end{enumerate}
\end{lemma}

\begin{proof}
For (1) $\Leftrightarrow$ (2) see \cite{Sidibe}. Noting that
$\mathbb{E}[X_t] = t \mathbb{E}[X_1]$ for all $t \geq 0$, the
equivalence (2) $\Leftrightarrow$ (3) follows from \cite[Prop.
3.17]{Cont-Tankov}.
\end{proof}

\begin{lemma}\label{lemma-exp-Levy-mt}
Let $X$ be a real-valued L\'evy process. The following statements
are equivalent:
\begin{enumerate}
\item $e^X$ is a local martingale;

\item $e^X$ is a martingale;

\item $\mathbb{E}[e^{X_1}] = 1$.
\end{enumerate}
\end{lemma}

\begin{proof}
The equivalence (1) $\Leftrightarrow$ (2) follows from \cite[Lemma
4.4.3]{Kallsen}. Noting that $\mathbb{E}[e^{X_t}] =
(\mathbb{E}[e^{X_1}])^t$ for all $t \geq 0$, the equivalence (2)
$\Leftrightarrow$ (3) follows from \cite[Prop. 3.17]{Cont-Tankov}.
\end{proof}

\section{Stock price models driven by bilateral Gamma
processes}\label{sec-bilateral}

In this section, we shall introduce bilateral Gamma stock models.
For this purpose, we first review the family of bilateral Gamma
processes. For details and more information, we refer to
\cite{Kuechler-Tappe}.

A \textit{bilateral Gamma distribution} with parameters $\alpha^+,
\lambda^+ ,\alpha^-, \lambda^- > 0$ is defined as the distribution
of $Y-Z$, where $Y$ and $Z$ are independent, $Y \sim
\Gamma(\alpha^+,\lambda^+)$ and $Z \sim \Gamma(\alpha^-,\lambda^-)$. For $\alpha,\lambda > 0$ we denote by
$\Gamma(\alpha,\lambda)$ a Gamma distribution, i.e. the absolutely continuous probability distribution with density
\begin{align*}
f(x) = \frac{\lambda^{\alpha}}{\Gamma(\alpha)} x^{\alpha-1} e^{-\lambda x} \mathbbm{1}_{(0,\infty)}(x), \quad x \in \mathbb{R}.
\end{align*}
The characteristic function of a bilateral Gamma distribution is
given by
\begin{align}\label{cf-bil-Gamma}
\phi(z) = \left( \frac{\lambda^+}{\lambda^+ - iz} \right)^{\alpha^+}
\left( \frac{\lambda^-}{\lambda^- + iz} \right)^{\alpha^-}, \quad z
\in \mathbb{R}
\end{align}
where the powers stem from the main branch of the complex logarithm.

Thus, any bilateral Gamma distribution is infinitely divisible,
which allows us to define its associated L\'evy process $(X_t)_{t
\geq 0}$, which we call a \textit{bilateral Gamma process}. We write
$X = (X_t)_{t \geq 0} \sim \Gamma(\alpha^+,\lambda^+;\alpha^-,\lambda^-)$ if $X_1$ has
a bilateral Gamma distribution with parameters $\alpha^+, \lambda^+
,\alpha^-, \lambda^- > 0$. All increments of $X$ have a bilateral
Gamma distribution, more precisely
\begin{align}\label{distribution-bil-time}
X_t - X_s \sim \Gamma(\alpha^+ (t-s), \lambda^+; \alpha^- (t-s),
\lambda^-) \quad \text{for $0 \leq s < t$.}
\end{align}
The characteristic triplet with respect to the truncation function
$h \equiv 0$ is given by $(0,0,F)$, where $F$ denotes the L\'evy
measure
\begin{align}\label{levy-measure}
F(dx) = \left( \frac{\alpha^+}{x} e^{-\lambda^+ x}
\mathbbm{1}_{(0,\infty)}(x) + \frac{\alpha^-}{|x|} e^{-\lambda^-
|x|} \mathbbm{1}_{(-\infty,0)}(x) \right)dx.
\end{align}
The \textit{cumulant generating function} $\Psi(z) = \ln \mathbb{E}
[ e^{zX_1} ]$ exists on $(-\lambda^-,\lambda^+)$ and is given by
\begin{align}\label{cumulant}
\Psi(z) &= \alpha^+ \ln \left( \frac{\lambda^+}{\lambda^+ - z}
\right) + \alpha^- \ln \left( \frac{\lambda^-}{\lambda^- + z}
\right), \quad z \in (-\lambda^-, \lambda^+).
\end{align}
We can write $X = X^+ - X^-$ as the difference of two independent
standard Gamma processes, where $X^+ \sim
\Gamma(\alpha^+,\lambda^+)$ and $X^- \sim
\Gamma(\alpha^-,\lambda^-)$. The corresponding cumulant generating
functions $\Psi^+(z) = \ln \mathbb{E} [ e^{zX_1^+} ]$ and $\Psi^-(z)
= \ln \mathbb{E} [ e^{zX_1^-} ]$ are given by
\begin{align}\label{cumulant-plus}
\Psi^+(z) &= \alpha^+ \ln \left( \frac{\lambda^+}{\lambda^+ - z}
\right), \quad z \in (-\infty,\lambda^+)
\\ \label{cumulant-minus} \Psi^-(z) &= \alpha^- \ln \left( \frac{\lambda^-}{\lambda^- - z} \right), \quad
z \in (-\infty, \lambda^-).
\end{align}
Note that $\Psi(z) = \Psi^+(z) + \Psi^-(-z)$ for $z \in (-\lambda^-,
\lambda^+)$.

A \textit{bilateral Gamma stock model} is an exponential L\'evy
model of the type (\ref{exp-Levy-model}) with $X$ being a bilateral
Gamma process. In what follows, we assume that $r \geq q \geq 0$,
that is, the dividend rate $q$ of the stock cannot exceed the
interest rate $r$ of the bank account and none of them is negative.

\begin{lemma}\label{lemma-martingale-measure}
Let $X \sim \Gamma(\alpha^+,\lambda^+;\alpha^-,\lambda^-)$ under
$\mathbb{P}$.

\begin{enumerate}
\item If $\lambda^+ > 1$, then $\mathbb{P}$ is a martingale
measure for $\tilde{S}$ if and only if
\begin{align}\label{martingale-eqn}
\left( \frac{\lambda^+}{\lambda^+ - 1} \right)^{\alpha^+} \left(
\frac{\lambda^-}{\lambda^- + 1} \right)^{\alpha^-} = e^{r-q}.
\end{align}

\item If $\lambda^+ \leq 1$, then $\mathbb{P}$ is never a martingale measure for $\tilde{S}$.

\end{enumerate}
\end{lemma}

\begin{proof}
If $\lambda^+ > 1$, we have $\mathbb{E}[e^{X_1}] < \infty$. By Lemma
\ref{lemma-exp-Levy-mt} the discounted price process $\tilde{S}$ in
(\ref{discounted-price}) is a local martingale if and only if
$\mathbb{E}[e^{X_1 - (r-q)}] = 1$, which is the case if and only if
(\ref{martingale-eqn}) holds.

In the case $\lambda^+ \leq 1$ we have $\mathbb{E}[e^{X_1}] =
\infty$. Lemma \ref{lemma-exp-Levy-mt} implies that $\tilde{S}$
cannot be a local martingale.
\end{proof}

\section{Existence of Esscher martingale measures in bilateral Gamma stock
models}\label{sec-Esscher}

For option pricing in bilateral Gamma stock models of the type
(\ref{exp-Levy-model}) we have to find a martingale measure. One
method is to use the so-called \textit{Esscher transform}, which was
pioneered in \cite{Gerber}. We recall the definition in the context
of bilateral Gamma processes.

\begin{definition}
Let $X \sim \Gamma(\alpha^+,\lambda^+;\alpha^-,\lambda^-)$ under
$\mathbb{P}$ and let $\Theta \in (-\lambda^-,\lambda^+)$ be
arbitrary. The {\rm Esscher transform} $\mathbb{P}^{\Theta} \overset{\rm loc}{\sim} \mathbb{P}$ is
defined as the locally equivalent probability measure with
likelihood process
\begin{align*}
\Lambda_t(\mathbb{P}^{\Theta},\mathbb{P}) := \frac{d
\mathbb{P}^{\Theta}}{d \mathbb{P}}\bigg|_{\mathcal{F}_t} = e^{\Theta
X_t - \Psi(\Theta) t}, \quad t \geq 0
\end{align*}
where $\Psi$ denotes the cumulant generating function given by
(\ref{cumulant}).
\end{definition}

\begin{lemma}\label{lemma-bilateral-preserving}
Let $X \sim \Gamma(\alpha^+,\lambda^+;\alpha^-,\lambda^-)$ under
$\mathbb{P}$ and let $\Theta \in (-\lambda^-,\lambda^+)$ be
arbitrary. Then we have $X \sim \Gamma(\alpha^+,\lambda^+ -
\Theta;\alpha^-,\lambda^- + \Theta)$ under $\mathbb{P}^{\Theta}$.
\end{lemma}

\begin{proof}
This follows from Proposition 2.1.3 and Example 2.1.4 in \cite{Kuechler-Soerensen}.
\end{proof}

The upcoming result characterizes all bilateral Gamma stock models,
for which martingale measures given by an Esscher transform exist.

\begin{theorem}\label{thm-Esscher}
Let $X \sim \Gamma(\alpha^+,\lambda^+;\alpha^-,\lambda^-)$ under
$\mathbb{P}$. Then there exists $\Theta \in (-\lambda^-,\lambda^+)$
such that $\mathbb{P}^{\Theta}$ is a martingale measure if and only
if
\begin{align}\label{cond-for-Esscher-lambda}
\lambda^+ + \lambda^- > 1.
\end{align}
If (\ref{cond-for-Esscher-lambda}) is satisfied, $\Theta$ is unique,
belongs to the interval $(-\lambda^-, \lambda^+ - 1)$, and it is the
unique solution of the equation
\begin{align}\label{Esscher-equation}
\left( \frac{\lambda^+ - \Theta}{\lambda^+ - \Theta - 1}
\right)^{\alpha^+} \left( \frac{\lambda^- + \Theta}{\lambda^- +
\Theta + 1} \right)^{\alpha^-} = e^{r-q}, \quad \Theta \in
(-\lambda^-, \lambda^+ - 1).
\end{align}
Moreover, we have $X \sim \Gamma(\alpha^+,\lambda^+ -
\Theta;\alpha^-,\lambda^- + \Theta)$ under $\mathbb{P}^{\Theta}$.
\end{theorem}

\begin{proof}
Let $\Theta \in (-\lambda^-,\lambda^+)$ be arbitrary. In view of
Lemma \ref{lemma-bilateral-preserving} and Lemma
\ref{lemma-martingale-measure}, the probability measure
$\mathbb{P}^{\Theta}$ is a martingale measure if and only if
$\lambda^+ - \Theta > 1$, i.e. $\Theta \in (-\lambda^-,\lambda^+ -
1)$, and (\ref{Esscher-equation}) is fulfilled. Note that
$(-\lambda^-,\lambda^+ - 1) \neq \emptyset$ if and only if
(\ref{cond-for-Esscher-lambda}) is satisfied.

Provided (\ref{cond-for-Esscher-lambda}), equation
(\ref{Esscher-equation}) is satisfied if and only if
\begin{align}\label{equation-theta}
f(\Theta) = r-q,
\end{align}
where $f : (-\lambda^-,\lambda^+ - 1) \rightarrow \mathbb{R}$ is
defined as $f(\Theta) := f^+(\Theta) + f^-(\Theta)$ with
\begin{align*}
f^+(\Theta) &:= \alpha^+
(\ln(\lambda^+ - \Theta) - \ln(\lambda^+ -1 - \Theta)),
\\ f^-(\Theta) &:= \alpha^- (\ln(\lambda^- + \Theta) - \ln(\lambda^- + 1 +
\Theta)).
\end{align*}
Taking derivatives of $f^+$ and $f^-$, we get that $f$ is strictly increasing.
Noting that
\begin{align*}
\lim_{\Theta \downarrow -\lambda^-} f(\Theta) = -\infty \quad
\text{and} \quad \lim_{\Theta \uparrow \lambda^+ - 1} f(\Theta) =
\infty,
\end{align*}
there exists a unique $\Theta \in (-\lambda^-,\lambda^+ - 1)$
fulfilling (\ref{equation-theta}).
\end{proof}

Hence, a martingale measure, which is an Esscher transform
$\mathbb{P}^{\Theta}$, exists if and only if
(\ref{cond-for-Esscher-lambda}) is satisfied. In order to find the
parameter $\Theta \in (-\lambda^-, \lambda^+ - 1)$, we have to solve
equation (\ref{Esscher-equation}). In general, this has to be done
numerically. In the particular situation $\alpha^+ = \alpha^-$,
which according to \cite[Thm. 3.3]{Kuechler-Tappe} is the case if
and only if $X$ is Variance Gamma, and $r=q$ the solution for
equation (\ref{Esscher-equation}) is given by $\Theta =
\frac{1}{2}(\lambda^+ - \lambda^- - 1)$, the midpoint of the
interval $(-\lambda^-, \lambda^+ - 1)$.

\begin{remark}
For general tempered stable distributions condition (\ref{cond-for-Esscher-lambda}) alone is
not sufficient for the existence of an Esscher martingale measure. This is due to the fact that the
cumulant generating functions of tempered stable distributions, which are not bilateral Gamma, have
finite values at the boundaries, which gives rise to an extra condition on the parameters.
\end{remark}

\section{Existence of minimal entropy martingale measures in bilateral Gamma stock
models}\label{sec-minimal-entropy}

We have seen in the previous section that in bilateral Gamma stock
models an equivalent martingale measure is easy to obtain by solving
equation (\ref{Esscher-equation}), provided condition
(\ref{cond-for-Esscher-lambda}) is satisfied. Moreover, the driving
process $X$ is still a bilateral Gamma process under the new
measure, which allows numerical calculations of option prices.
However, it is not clear that, in reality, the market chooses this
kind of measure.

In the literature, one often performs option pricing by finding an
equivalent martingale measure $\mathbb{Q} \sim \mathbb{P}$ which
minimizes
\begin{align*}
\mathbb{E}[ f (\Lambda_1(\mathbb{Q},\mathbb{P})) ]
\end{align*}
for a strictly convex function $f : (0,\infty) \rightarrow
\mathbb{R}$. Popular choices for the functional $f$ are $f(x) = x^p$ for
some $p > 1$ (see, e.g., \cite{Takuji-p, Hobson, Santacroce-2006, Kohlmann-2007, Jeanblanc, Bender}) and $f(x) = x \ln
x$. In the special case $p=2$ the measure $\mathbb{Q}$ is known as the 
\textit{variance-optimal martingale measure}, see, e.g., \cite{Schweizer-VOMM, Delb-Schach-VOMM, Mania-1, Mania-2, Takuji, Cerny-VOMM}.
It follows from \cite[Ex. 2.7]{Bender} that for bilateral Gamma stock
models a $p$-optimal equivalent martingale measure does not exist, because we would need that the tails for upward jumps are extraordinarily light. We can merely obtain
the existence of a $p$-optimal signed martingale measure and of a $p$-optimal absolutely continuous martingale measure for our model.

In this section, we shall consider the
functional $f(x) = x \ln x$. Then, the quantity
\begin{align*}
\mathbb{H}(\mathbb{Q} \,|\, \mathbb{P}) := \mathbb{E}
[\Lambda_1(\mathbb{Q},\mathbb{P}) \ln
\Lambda_1(\mathbb{Q},\mathbb{P})] = \mathbb{E}_{\mathbb{Q}}[\ln
\Lambda_1(\mathbb{Q},\mathbb{P})]
\end{align*}
is called the \textit{relative entropy}. In connection with
exponential L\'evy models, it has been studied, e.g., in
\cite{Chan,Miyahara,Esche-Schweizer,Hubalek}, see also \cite{Krol}.
The minimal entropy has an information theoretic interpretation:
minimizing relative entropy corresponds to choosing a martingale
measure by adding the least amount of information to the prior
model.

From a mathematical point of view, minimizing the relative entropy
is convenient, because there is a connection between the minimal
entropy and Esscher transforms of the exponential transform
$\tilde{Y} := \mathcal{L}(\tilde{S})$ of the discounted stock price process $\tilde{S}$, which has rigorously been presented in
\cite{Esche-Schweizer} and further been extended in \cite{Hubalek}. Note that $\tilde{S}_t = S_0 e^{Y_t}$, where $Y$ denotes
the L\'evy process $Y_t := X_t - (r-q)t$.

We shall now recall this connection in our framework, where $X$ is a
bilateral Gamma process. By \cite[Thm. 2]{Hubalek}, the exponential
transform $\tilde{Y}$ is a L\'evy process, which has the
characteristic triplet
\begin{align*}
\tilde{b} &= \int_{\mathbb{R}} h(e^x - 1) F(dx) - (r-q),
\\ \tilde{c} &= 0,
\\ \tilde{F}(B) &= \int_{\mathbb{R}} \mathbbm{1}_B (e^x - 1) F(dx), \quad B \in \mathcal{B}(\mathbb{R}).
\end{align*}
with respect to the truncation function $h(x) = x
\mathbbm{1}_{[-1,1]}(x)$, where the L\'evy measure $F$ is given by
(\ref{levy-measure}). Therefore, the cumulant generating function
$\tilde{\Psi}$ of $\tilde{Y}$ exists on $\mathbb{R}_- :=
(-\infty,0]$ and is given by
\begin{equation}\label{cumulant-tilde}
\begin{aligned}
\tilde{\Psi}(z) &= -(r-q)z + \int_{\mathbb{R}} ( e^{zx} - 1 )
\tilde{F}(dx)
\\ &= -(r-q)z + \int_{\mathbb{R}} \left( e^{z(e^x -
1)} - 1 \right) F(dx), \quad z \leq 0.
\end{aligned}
\end{equation}

\begin{definition}
Let $\vartheta \leq 0$ be arbitrary. The {\rm Esscher transform}
$\mathbb{P}_{\vartheta} \overset{\rm loc}{\sim} \mathbb{P}$ is the locally equivalent probability
measure with likelihood process
\begin{align*}
\Lambda_t(\mathbb{P}_{\vartheta},\mathbb{P}) := \frac{d
\mathbb{P}_{\vartheta}}{d \mathbb{P}}\bigg|_{\mathcal{F}_t} =
e^{\vartheta \tilde{Y}_t - \tilde{\Psi}(\vartheta) t}, \quad t \geq
0.
\end{align*}
\end{definition}

The following result, which establishes the connection between the
minimal entropy and Esscher transforms, will be crucial for our
investigations.

\begin{proposition}\label{prop-connection}
There exists a minimal entropy martingale measure if and only if
there exists $\vartheta \leq 0$ such that
$\mathbb{E}_{\vartheta}[\tilde{Y}_1] = 0$. In this case, one minimal
entropy martingale measure is given by $\mathbb{P}_{\vartheta}$.
\end{proposition}

\begin{proof}
The assertion follows from \cite[Thm. 8]{Hubalek} and Lemma
\ref{lemma-Levy-mt}.
\end{proof}


We are now ready to characterize all bilateral Gamma stock models,
for which minimal entropy martingale measures exist, and to determine the value of this minimal entropy.

\begin{theorem}\label{thm-MEMM}
Let $X \sim \Gamma(\alpha^+,\lambda^+;\alpha^-,\lambda^-)$ under
$\mathbb{P}$.
\begin{enumerate}
\item If $\lambda^+ \leq 1$, there exists a unique $\vartheta < 0$ such that $\mathbb{P}_{\vartheta}$ is a minimal entropy
martingale measure. It is the unique solution of the equation
\begin{equation}\label{eqn-minimize-entropy}
\begin{aligned}
&\alpha^+ \int_0^{\infty} \frac{1}{x} e^{-\lambda^+ x} (e^x - 1)
e^{\vartheta(e^x - 1)} dx
\\ &+ \alpha^- \int_0^{\infty} \frac{1}{x}
e^{-\lambda^- x} (e^{-x} - 1) e^{\vartheta(e^{-x} - 1)} dx = r-q,
\quad \vartheta \in (-\infty,0).
\end{aligned}
\end{equation}

\item If $\lambda^+ > 1$, there exists $\vartheta \leq 0$ such that $\mathbb{P}_{\vartheta}$ is a minimal entropy
measure if and only if
\begin{align}\label{entropy-inequality}
\alpha^+ \ln \left( \frac{\lambda^+}{\lambda^+ - 1} \right) + \alpha^- \ln \left( \frac{\lambda^-}{\lambda^- + 1} \right) \geq r-q.
\end{align}
If (\ref{entropy-inequality}) is satisfied, $\vartheta$ is unique
and it is the unique solution of equation
(\ref{eqn-minimize-entropy}) for $\vartheta \in \mathbb{R}_-$.

\item If a minimal entropy martingale measure exists, then the value
of the minimal entropy is given by

\begin{equation}\label{value-entropy}
\begin{aligned}
\mathbb{H}(\mathbb{P}_{\vartheta} \,|\, \mathbb{P}) &= - \alpha^+
\int_0^{\infty} \frac{1}{x} e^{-\lambda^+ x} (e^{\vartheta(e^x - 1)}
- 1) dx
\\ &\quad - \alpha^- \int_0^{\infty} \frac{1}{x} e^{-\lambda^- x}
(e^{\vartheta(e^{-x} - 1)} - 1) dx + (r-q) \vartheta.
\end{aligned}
\end{equation}
\end{enumerate}
\end{theorem}

\begin{proof}
For each $\vartheta \leq 0$ the characteristic triplet of
$\tilde{Y}$ with respect to the truncation function $h(x) = x
\mathbbm{1}_{[-1,1]}(x)$ under the measure $\mathbb{P}_{\vartheta}$
is, according to \cite[Thm. 1]{Hubalek}, given by
\begin{align*}
\tilde{b}_{\vartheta} &= \tilde{b} + \int_{\mathbb{R}} (e^{\vartheta
x} - 1)h(x) \tilde{F}(dx) = \int_{\mathbb{R}} h(e^x - 1) e^{\vartheta(e^x -
1)} F(dx) - (r-q),
\\ \tilde{c}_{\vartheta} &= 0,
\\ \tilde{F}_{\vartheta}(dx) &= e^{\vartheta x} \tilde{F}(dx).
\end{align*}
Hence, we have $\mathbb{E}_{\vartheta}|\tilde{Y}_1| < \infty$,
$\vartheta < 0$ and, by (\ref{levy-measure}), the expectations are
given by
\begin{align*}
\mathbb{E}_{\vartheta}[\tilde{Y}_1] &= \int_{\mathbb{R}} (e^x - 1)
e^{\vartheta(e^x - 1)} F(dx) - (r-q)
\\ &= \alpha^+ \int_0^{\infty} \frac{1}{x} e^{-\lambda^+ x} (e^x - 1) e^{\vartheta(e^x - 1)} dx
\\ &\quad + \alpha^- \int_0^{\infty} \frac{1}{x} e^{-\lambda^- x} (e^{-x} -
1) e^{\vartheta(e^{-x} - 1)} dx - (r-q), \quad \vartheta \leq 0.
\end{align*}
Moreover, we have $\mathbb{E}_0 |\tilde{Y}_1| < \infty$ if and only
if $\lambda^+ > 1$, and in this case the expectation is, by (\ref{cumulant}), given by
\begin{align}\label{boundary}
\mathbb{E}_0 [\tilde{Y}_1] = \Psi(1) - (r-q) = \alpha^+ \ln \left(
\frac{\lambda^+}{\lambda^+ - 1} \right) + \alpha^- \ln \left(
\frac{\lambda^-}{\lambda^- + 1} \right) - (r-q)
\end{align}
According to \cite[Lemma 26.4]{Sato} the cumulant generating
function $\tilde{\Psi}$ of $\tilde{Y}$ is of class $C^{\infty}$ on
$(-\infty,0)$, we have $\tilde{\Psi}'' > 0$ on $(-\infty,0)$ and the
first derivative is, by the representation (\ref{cumulant-tilde}),
given by
\begin{align*}
\tilde{\Psi}'(\vartheta) = -(r-q) + \int_{\mathbb{R}} (e^x - 1)
e^{\vartheta(e^x - 1)} F(dx) = \mathbb{E}_{\vartheta}[\tilde{Y}_1],
\quad \vartheta \in (-\infty,0).
\end{align*}
Note that $\tilde{\Psi}'(\vartheta) \downarrow - \infty$ for
$\vartheta \downarrow -\infty$.

Let us now consider the situation $\lambda^+ \leq 1$. Then we have
$\tilde{\Psi}'(\vartheta) \uparrow \infty$ for $\vartheta \uparrow
0$. Since $\tilde{\Psi}'$ is continuous and $\tilde{\Psi}'' > 0$ on
$(-\infty,0)$, Proposition \ref{prop-connection} yields the first
statement.

In the case $\lambda^+ > 1$, an analogous argumentation yields, by
taking into account (\ref{boundary}), the second statement.

If a minimal entropy martingale measure exists, then the value of
the minimal entropy is given by
\begin{align*}
\mathbb{H}(\mathbb{P}_{\vartheta} \,|\, \mathbb{P}) =
\mathbb{E}_{\vartheta}[\ln
\Lambda_1(\mathbb{P}_{\vartheta},\mathbb{P})] =
\mathbb{E}_{\vartheta}[\vartheta \tilde{Y}_1 -
\tilde{\Psi}(\vartheta)] = -\tilde{\Psi}(\vartheta),
\end{align*}
which, by (\ref{cumulant-tilde}) and (\ref{levy-measure}), gives us
(\ref{value-entropy}).
\end{proof}

We emphasize that for bilateral Gamma stock models the value of the minimal entropy in (\ref{value-entropy}) can easily be calculated numerically in terms of the parameters $\alpha^+,\lambda^+,\alpha^-,\lambda^-$.

\begin{remark}
For tempered stable processes a similar version of Theorem \ref{thm-MEMM} holds true. The terms in (\ref{eqn-minimize-entropy}) and (\ref{value-entropy}) slightly change, namely, the fraction $\frac{1}{x}$ is replaced by $\frac{1}{x^{1 + \beta^+}}$ resp. $\frac{1}{x^{1 + \beta^-}}$, where $\beta^+,\beta^- \in (0,1)$ are the additional two parameters. Condition (\ref{entropy-inequality}) differs considerably in the tempered stable case, due to the different form of the cumulant generating function $\Psi$.
\end{remark}

\section{Existence of minimal entropy martingale measures preserving the class of bilateral Gamma
processes}\label{sec-bilateral-Esscher}

We have seen in the previous section that for bilateral Gamma stock
models we obtain the minimal entropy martingale measure
$\mathbb{P}_{\vartheta}$ by solving equation
(\ref{eqn-minimize-entropy}) numerically, provided $\lambda^+ \leq
1$ or condition (\ref{entropy-inequality}) is satisfied. Under
$\mathbb{P}_{\vartheta}$, the bilateral Gamma process $X$ is still a
L\'evy process, a result which is due to \cite{Esche-Schweizer}, and
we know its characteristic triplet. However, neither its one-dimensional density nor its characteristic
function is available in closed form, hence we cannot perform
option pricing numerically.

Recall that, on the other hand, the Esscher transform from Section
\ref{sec-Esscher} leaves the family of bilateral Gamma processes
invariant.

Our idea in this section is therefore as follows. We minimize the
relative entropy within the class of bilateral Gamma processes by
performing \textit{bilateral} Esscher transforms.

\begin{definition}
Let $X \sim \Gamma(\alpha^+,\lambda^+;\alpha^-,\lambda^-)$ under
$\mathbb{P}$ and let $\theta^+ \in (-\infty,\lambda^+)$ and
$\theta^- \in (-\infty,\lambda^-)$ be arbitrary. The {\rm bilateral
Esscher transform} $\mathbb{P}^{(\theta^+,\theta^-)} \overset{\rm loc}{\sim} \mathbb{P}$ is defined as
the locally equivalent probability measure with likelihood process
\begin{align*}
\Lambda_t(\mathbb{P}^{(\theta^+,\theta^-)},\mathbb{P}) := \frac{d
\mathbb{P}^{(\theta^+,\theta^-)}}{d
\mathbb{P}}\bigg|_{\mathcal{F}_t} = e^{\theta^+ X_t^+ -
\Psi^+(\theta^+) t} \cdot e^{\theta^- X_t^- - \Psi^-(\theta^-) t},
\quad t \geq 0
\end{align*}
where $\Psi^+, \Psi^-$ denote the cumulant generating functions
given by (\ref{cumulant-plus}), (\ref{cumulant-minus}).
\end{definition}

Note that the Esscher transforms $\mathbb{P}^{\Theta}$ from Section
\ref{sec-Esscher} are special cases of the just introduced bilateral
Esscher transforms $\mathbb{P}^{(\theta^+,\theta^-)}$. Indeed, it
holds
\begin{align}\label{special-case}
\mathbb{P}^{\Theta} = \mathbb{P}^{(\Theta,-\Theta)}, \quad \Theta
\in (-\lambda^-,\lambda^+).
\end{align}

\begin{lemma}\label{lemma-two-sided-preserving}
Let $X \sim \Gamma(\alpha^+,\lambda^+;\alpha^-,\lambda^-)$ under
$\mathbb{P}$ and let $\theta^+ \in (-\infty,\lambda^+)$ and
$\theta^- \in (-\infty,\lambda^-)$ be arbitrary. Then we have $X
\sim \Gamma(\alpha^+,\lambda^+ - \theta^+;\alpha^-,\lambda^- -
\theta^-)$ under $\mathbb{P}^{(\theta^+,\theta^-)}$.
\end{lemma}

\begin{proof}
This follows from Proposition 2.1.3 and Example 2.1.4 in \cite{Kuechler-Soerensen}.
\end{proof}

According to \cite[Prop. 6.1]{Kuechler-Tappe}, the parameters $\alpha^+$ and $\alpha^-$ cannot be changed to
other parameters by an equivalent measure transformation. Consequently, any equivalent measure transformation, under which $X$ is still a bilateral Gamma process, is a bilateral Esscher transform.

\begin{lemma}\label{lemma-mm-bilateral}
Let $X \sim \Gamma(\alpha^+,\lambda^+;\alpha^-,\lambda^-)$ under
$\mathbb{P}$ and let $\theta^+ \in (-\infty,\lambda^+)$ and
$\theta^- \in (-\infty,\lambda^-)$ be arbitrary. Then
$\mathbb{P}^{(\theta^+,\theta^-)}$ is a martingale measure if and
only if $\theta^+ \in (\lambda^+ - (1- \exp(-\frac{r -
q}{\alpha^+}))^{-1},\lambda^+ - 1)$ and
\begin{align}\label{eqn-parameter}
\theta^- = \Phi(\theta^+),
\end{align}
where $\Phi : (\lambda^+ - (1- \exp(-\frac{r -
q}{\alpha^+}))^{-1},\lambda^+ - 1) \rightarrow (-\infty,\lambda^-)$
is defined as the strictly increasing function
\begin{align}\label{def-phi}
\Phi(\theta) := \lambda^- - \Bigg( \sqrt[\alpha^-]{\bigg(
\frac{\lambda^+ - \theta}{\lambda^+ - \theta - 1} \bigg)^{\alpha^+}
e^{- (r-q)}} - 1 \Bigg)^{-1}.
\end{align}
By convention, we set $\lambda^+ - (1- \exp(-\frac{r -
q}{\alpha^+}))^{-1} := -\infty$ if $r = q$.
\end{lemma}

\begin{proof}
This is an immediate consequence of Lemma
\ref{lemma-two-sided-preserving} and Lemma
\ref{lemma-martingale-measure}.
\end{proof}

As pointed out above, all equivalent measure transformations preserving the class of bilateral Gamma processes are bilateral Esscher transforms. Hence, we introduce the set of parameters
\begin{align*}
\mathcal{M}_{\mathbb{P}} := \{ (\theta^+, \theta^-) \in
(-\infty,\lambda^+) \times (-\infty,\lambda^-) \,|\,
\text{$\mathbb{P}^{(\theta^+,\theta^-)}$ is a martingale measure} \}
\end{align*}
such that the bilateral Esscher transform is a martingale measure. The previous Lemma \ref{lemma-mm-bilateral} tells us that
\begin{align}\label{set-MM}
\mathcal{M}_{\mathbb{P}} = \{ (\theta,\Phi(\theta)) \in \mathbb{R}^2
\,|\, \theta \in ( \lambda^+ - ( 1- \exp ( {\textstyle -\frac{r -
q}{\alpha^+}} ) )^{-1},\lambda^+ - 1 ) \}.
\end{align}
The following consequence contributes to Theorem \ref{thm-Esscher}.
It tells us that, provided (\ref{cond-for-Esscher-lambda}) is
satisfied, we can, instead of solving (\ref{Esscher-equation}),
alternatively solve equation (\ref{fixed-point-eqn}) below in order
to find the Esscher transform $\mathbb{P}^{\Theta}$.

\begin{corollary}\label{cor-Esscher-eqn}
Let $X \sim \Gamma(\alpha^+,\lambda^+;\alpha^-,\lambda^-)$ under
$\mathbb{P}$. If (\ref{cond-for-Esscher-lambda}) is satisfied, then
the unique $\Theta \in (-\lambda^-,\lambda^+ - 1)$ such that
$\mathbb{P}^{\Theta}$ is a martingale measure is the unique solution
of the equation
\begin{align}\label{fixed-point-eqn}
\Phi(\Theta) = -\Theta, \quad \Theta \in (\lambda^+ - (1-
\exp({\textstyle -\frac{r - q}{\alpha^+}}))^{-1},\lambda^+ - 1).
\end{align}
\end{corollary}

\begin{proof}
The proof follows from Theorem \ref{thm-Esscher} and relations
(\ref{special-case}), (\ref{set-MM}).
\end{proof}

We have considered the case $\alpha^+ = \alpha^-$ and $r=q$ at the
end of Section \ref{sec-Esscher}. In this particular situation $\Phi
: (-\infty,\lambda^+ - 1) \rightarrow (-\infty,\lambda^-)$ is given
by $\Phi(\theta) = \lambda^- - \lambda^+ + \theta + 1$, whence we
see again, this time by solving equation (\ref{fixed-point-eqn}),
that the Esscher parameter is given by $\Theta =
\frac{1}{2}(\lambda^+ - \lambda^- - 1)$.

We are now ready to treat the existence of minimal entropy
martingale measures in bilateral Gamma stock markets, under which
$X$ remains a bilateral Gamma process.

\begin{theorem}\label{thm-min-entropy-bilateral}
Let $X \sim \Gamma(\alpha^+,\lambda^+;\alpha^-,\lambda^-)$ under
$\mathbb{P}$ with arbitrary parameters
$\alpha^+,\lambda^+,\alpha^-,\lambda^- > 0$. Then, there exist
$\theta^+ \in (-\infty,\lambda^+)$ and $\theta^- \in
(-\infty,\lambda^-)$ such that
\begin{align}\label{min-entropy-bilateral}
\mathbb{H}(\mathbb{P}^{(\theta^+,\theta^-)} \,|\, \mathbb{P}) =
\min_{(\vartheta^+, \vartheta^-) \in \mathcal{M}_{\mathbb{P}}}
\mathbb{H}(\mathbb{P}^{(\vartheta^+,\vartheta^-)} \,|\, \mathbb{P}).
\end{align}
In this case, we have $\theta^+ \in (\lambda^+ - (1- \exp(-\frac{r -
q}{\alpha^+}))^{-1},\lambda^+ - 1)$, relation (\ref{eqn-parameter})
is satisfied, and $\theta^+$ minimizes the function $f : (\lambda^+
- (1- \exp(-\frac{r - q}{\alpha^+}))^{-1}, \lambda^+ - 1)
\rightarrow \mathbb{R}$ defined as
\begin{equation}\label{def-f}
\begin{aligned}
f(\theta) &:= \alpha^+ \left( \frac{\lambda^+}{\lambda^+ - \theta} -
1 - \ln \bigg( \frac{\lambda^+}{\lambda^+ - \theta} \bigg) \right)
\\ & \quad + \alpha^- \left( \frac{\lambda^-}{\lambda^- - \Phi(\theta)} - 1 -
\ln \bigg( \frac{\lambda^-}{\lambda^- - \Phi(\theta)} \bigg)
\right).
\end{aligned}
\end{equation}
Moreover, we have $X \sim \Gamma(\alpha^+,\lambda^+ -
\theta^+;\alpha^-,\lambda^- - \theta^-)$ under
$\mathbb{P}^{(\theta^+,\theta^-)}$, and the value of the minimal
entropy is given by
\begin{align}\label{compute-entropy}
\mathbb{H}(\mathbb{P}^{(\theta^+,\theta^-)} \,|\, \mathbb{P}) =
f(\theta^+).
\end{align}
\end{theorem}

\begin{proof}
For $\theta^+ \in (-\infty,\frac{\lambda^+}{p})$ and $\theta^- \in (-\infty,\frac{\lambda^-}{p})$ the relative entropy is, by taking into account Lemma \ref{lemma-two-sided-preserving}, given by
\begin{align*}
\mathbb{H}(\mathbb{P}^{(\theta^+,\theta^-)} \,|\, \mathbb{P}) &= \mathbb{E}_{\mathbb{P}^{(\theta^+,\theta^-)}} \bigg[ \ln \bigg( \frac{d \mathbb{P}^{(\theta^+,\theta^-)}}{d \mathbb{P}} \bigg) \bigg]
\\ &= \mathbb{E}_{\mathbb{P}^{(\theta^+,\theta^-)}} \big[ \theta^+ X_1^+ - \Psi^+(\theta^+) + \theta^- X_1^- - \Psi^-(\theta^-) \big]
\\ &= \frac{\theta^+ \alpha^+}{\lambda^+ - \theta^+} - \Psi^+(\theta^+) + \frac{\theta^- \alpha^-}{\lambda^- - \theta^-} - \Psi^-(\theta^-)
\\ &= \alpha^+ g \bigg( \frac{\lambda^+}{\lambda^+ - \theta^+}
\bigg) + \alpha^- g \bigg( \frac{\lambda^-}{\lambda^- -
\theta^-} \bigg),
\end{align*}
where $g : (0,\infty) \rightarrow \mathbb{R}$ denotes the strictly
convex function $g(x) = x - 1 - \ln x$. Thus, for all $(\theta^+,\theta^-)
\in \mathcal{M}_{\mathbb{P}}$ the relative entropy
$\mathbb{H}(\mathbb{P}^{(\theta^+,\theta^-)} \,|\, \mathbb{P})$ is
given by (\ref{compute-entropy}), and we can write the function $f$ as
\begin{align}\label{repr-f}
f(\theta) = \alpha^+ g \bigg( \frac{\lambda^+}{\lambda^+ - \theta}
\bigg) + \alpha^- g \bigg( \frac{\lambda^-}{\lambda^- -
\Phi(\theta)} \bigg).
\end{align}
Note that
\begin{align*}
\lim_{\theta \downarrow \lambda^+ - (1- \exp(-\frac{r -
q}{\alpha^+}))^{-1}} \frac{\lambda^+}{\lambda^+ - \theta} \in
[0,\lambda^+) \quad \text{and} \quad \lim_{\theta \uparrow \lambda^+
- 1} \frac{\lambda^+}{\lambda^+ - \theta} = \lambda^+
\end{align*}
as well as
\begin{align*}
\lim_{\theta \downarrow \lambda^+ - (1- \exp(-\frac{r -
q}{\alpha^+}))^{-1}} \frac{\lambda^-}{\lambda^- - \Phi(\theta)} = 0
\quad \text{and} \quad \lim_{\theta \uparrow \lambda^+ - 1}
\frac{\lambda^-}{\lambda^- - \Phi(\theta)} = \infty.
\end{align*}
We conclude that
\begin{align*}
\lim_{\theta \downarrow \lambda^+ - (1- \exp(-\frac{r -
q}{\alpha^+}))^{-1}} f(\theta) = \infty \quad \text{and} \quad
\lim_{\theta \uparrow \lambda^+ - 1} f(\theta) = \infty.
\end{align*}
Since $f$ is continuous, it attains a minimum and the assertion
follows.
\end{proof}

Consequently, unlike the Esscher transform $\mathbb{P}^{\theta}$
from Section \ref{sec-Esscher} and the real minimal entropy
martingale measure $\mathbb{P}_{\vartheta}$ from Section
\ref{sec-minimal-entropy}, the minimal entropy martingale measure
$\mathbb{P}^{(\theta,\Phi(\theta))}$ in the subclass of all measures
which leave bilateral Gamma processes invariant, \textit{always}
exists. In this case, it is also possible to determine the minimal entropy numerically
by minimizing the function $f$ defined in (\ref{def-f}). A comparison with the value in (\ref{value-entropy}) will be given in a concrete example in Section \ref{sec-numerics}.

\begin{remark}\label{remark-Esscher}
If condition (\ref{cond-for-Esscher-lambda}) is satisfied, choosing
the Esscher parameter $\Theta \in (-\lambda^-,\lambda^+ - 1)$ from
Section \ref{sec-Esscher} and inserting it into the function $f$, we
obtain, by Corollary \ref{cor-Esscher-eqn} and the representation
(\ref{repr-f}) of $f$,
\begin{align}\label{insert-Esscher}
\mathbb{H}(\mathbb{P}^{\Theta} \,|\, \mathbb{P}) = f(\Theta) =
\alpha^+ g \bigg( \frac{\lambda^+}{\lambda^+ - \Theta} \bigg) +
\alpha^- g \bigg( \frac{\lambda^-}{\lambda^- + \Theta} \bigg).
\end{align}
Note that one of the arguments in (\ref{insert-Esscher}) for the
function $g$ is greater than $1$, whereas the other argument is
smaller than $1$. Since the strictly convex function $g$ attains its global minimum at $x = 1$, and, by Taylor's theorem, behaves like $(x-1)^2$ around $x=1$, the relative entropy
$\mathbb{H}(\mathbb{P}^{\Theta} \,|\, \mathbb{P})$ in
(\ref{insert-Esscher}) is, in general, not too far from the minimal relative
entropy $\mathbb{H}(\mathbb{P}^{(\theta,\Phi(\theta))} \,|\,
\mathbb{P})$ in (\ref{min-entropy-bilateral}). Therefore, we expect
that, typically, the Esscher parameter $\Theta$ and the bilateral
Esscher parameter $\theta$ are close to each other.
\end{remark}

In general, we have the inequalities
\begin{align*}
\mathbb{H}(\mathbb{P}_{\vartheta} \,|\, \mathbb{P}) <
\mathbb{H}(\mathbb{P}^{(\theta,\Phi(\theta))} \,|\, \mathbb{P}) <
\mathbb{H}(\mathbb{P}^{\Theta} \,|\, \mathbb{P}),
\end{align*}
provided, the respective measures exist. Using our preceding
results, we can compute the values of the respective entropies
numerically and see, how close they are to each other.

Since $X$ is still a bilateral Gamma process under the bilateral
Esscher transform $\mathbb{P}^{(\theta,\Phi(\theta))}$, we can
perform option pricing by the method of Fourier transformation.

The characteristic function of a bilateral Gamma distribution is
given by (\ref{cf-bil-Gamma}) and all increments of $X$ are
bilateral Gamma distributed, see (\ref{distribution-bil-time}).
Therefore, if $\lambda^+ - \theta
> 1$, the price of a European call option with strike price $K$ and
time of maturity $T$ is given by
\begin{equation}\label{option-pricing}
\begin{aligned}
\pi(S_0) &= - \frac{e^{-rT} K}{2 \pi} \int_{i \nu - \infty}^{i \nu +
\infty} \left( \frac{K}{S_0} \right)^{iz} \bigg( \frac{\lambda^+ -
\theta}{\lambda^+ - \theta + iz} \bigg)^{\alpha^+ T} 
\\ &\qquad\qquad\qquad \times \bigg(
\frac{\lambda^- - \Phi(\theta)}{\lambda^- - \Phi(\theta) - iz}
\bigg)^{\alpha^- T} \frac{dz}{z(z-i)},
\end{aligned}
\end{equation}
where $\nu \in (1,\lambda^+ - \theta)$ is arbitrary. This follows
from \cite[Thm. 3.2]{Lewis}.

As mentioned in Section \ref{sec-minimal-entropy}, for bilateral Gamma stock models the $p$-optimal martingale
measure does not exist. However, we can, as provided for the minimal entropy martingale measure, determine
the $p$-optimal martingale measure within the class of bilateral Gamma processes.

The $p$-distance of a bilateral Esscher transform is easy to compute. Indeed, since $X^+$ and $X^-$ are independent, for $p > 1$ and $\theta^+ \in (-\infty,\frac{\lambda^+}{p})$, $\theta^- \in (-\infty,\frac{\lambda^-}{p})$ the $p$-distance is given by
\begin{equation}\label{p-distance}
\begin{aligned}
&\mathbb{E} \left[ \bigg( \frac{d \mathbb{P}^{(\theta^+,\theta^-)}}{d \mathbb{P}} \bigg)^p \right] = e^{-p(\Psi^+(\theta^+) + \Psi^-(\theta^-))} \mathbb{E} \big[ e^{p \theta^+ X_1^+} \big] \mathbb{E} \big[ e^{p \theta^- X_1^-} \big]
\\ &= \exp \big( -p (\Psi^+(\theta^+) + \Psi^-(\theta^-)) + \Psi^+(p\theta^+) + \Psi^-(p \theta^-) \big)
\\ &= \bigg( \frac{\lambda^+}{\lambda^+ - p\theta^+} \bigg)^{\alpha^+} \bigg( \frac{\lambda^-}{\lambda^- - p \theta^-} \bigg)^{\alpha^-} \bigg( \frac{\lambda^+ - \theta^+}{\lambda^+} \bigg)^{p \alpha^+} \bigg( \frac{\lambda^- - \theta^-}{\lambda^-} \bigg)^{p \alpha^-}.
\end{aligned}
\end{equation}
A similar argumentation as in Theorem \ref{thm-min-entropy-bilateral} shows that, for each choice of the parameters $\alpha^+,\lambda^+,\alpha^-,\lambda^- > 0$ there exists a pair $(\theta^+,\theta^-)$ minimizing the $p$-distance (\ref{p-distance}), and in this case we also have (\ref{eqn-parameter}), where $\theta^+$ minimizes the function
\begin{align}\label{def-f-p}
f_p(\theta) = \bigg( \frac{\lambda^+}{\lambda^+ - p\theta} \bigg)^{\alpha^+} \bigg( \frac{\lambda^-}{\lambda^- - p \Phi(\theta)} \bigg)^{\alpha^-} \bigg( \frac{\lambda^+ - \theta}{\lambda^+} \bigg)^{p \alpha^+} \bigg( \frac{\lambda^- - \Phi(\theta)}{\lambda^-} \bigg)^{p \alpha^-}.
\end{align}
As we shall see in the numerical illustration of the upcoming section, we have $\theta_p \rightarrow \theta$ for $p \downarrow 1$, where for each $p > 1$ the parameter $\theta_p$ minimizes (\ref{def-f-p}) and $\theta$ minimizes (\ref{def-f}). This is not surprising, since it is known that, under suitable technical conditions, the $p$-optimal martingale measure converges to the minimal entropy martingale measure for $p \downarrow 1$, see, e.g., \cite{Grandits, Grandits-R, Santacroce-2005, Jeanblanc, Bender, Kohlmann-2008}. Here, we shall only give an intuitive argument. For $\alpha,\lambda > 0$ consider the functions
\begin{align*}
h(\theta) &= \alpha \bigg[ \frac{\lambda}{\lambda - \theta} - 1 - \ln \bigg( \frac{\lambda}{\lambda - \theta} \bigg) \bigg],
\\ h_p(\theta) &= \bigg( \frac{\lambda}{\lambda - p\theta} \bigg)^{\alpha} \bigg( \frac{\lambda - \theta}{\lambda} \bigg)^{p \alpha} \quad \text{for $p > 1$.}
\end{align*}
Note that $h$ and $h_p$ have a global minimum at $\theta = 0$. Moreover,
if $p > 1$ is close to $1$, then the functions $(p-1) h(\theta)$ and $\ln h_p(\theta)$ are very close to each other. Indeed, Taylor's theorem shows that, for an appropriate $\xi \in \mathbb{R}$, which is between $1$ and $\frac{\lambda - p\theta}{\lambda - \theta}$, we have
\begin{align*}
&| (p-1) h(\theta) - \ln h_p(\theta) |
\\ &= \alpha \bigg| (p-1) \bigg[ \frac{\lambda}{\lambda - \theta} - 1 - \ln \bigg( \frac{\lambda}{\lambda - \theta} \bigg) \bigg] - \ln \bigg( \frac{\lambda}{\lambda - p \theta} \bigg) + p \ln \bigg( \frac{\lambda}{\lambda - \theta} \bigg) \bigg|
\\ &= \alpha \bigg| \frac{(p-1)\theta}{\lambda - \theta} + \ln \bigg( \frac{\lambda - p\theta}{\lambda - \theta} \bigg) \bigg|
\\ &= \alpha \bigg| \frac{(p-1)\theta}{\lambda - \theta} + \bigg[ \frac{\lambda - p\theta}{\lambda - \theta} - 1 - \frac{1}{2 \xi^2} \bigg( \frac{\lambda - p\theta}{\lambda - \theta} - 1 \bigg)^2 \bigg] \bigg|
= \frac{\alpha}{2 \xi^2} \bigg( \frac{\lambda - p\theta}{\lambda - \theta} - 1 \bigg)^2.
\end{align*}
For $p > 1$ close to $1$, the expression $\frac{\lambda - p\theta}{\lambda - \theta}$ is close to $1$ for all relevant $\theta$. Intuitively, this explains the convergence $\theta_p \rightarrow \theta$ for $p \downarrow 1$, which implies the convergence of the $p$-optimal martingale measure to the minimal entropy martingale measure.

\begin{remark}
For tempered stable processes, which are not bilateral Gamma, the situation in this section becomes more involved. It may happen that $\mathcal{M}_{\mathbb{P}} = \emptyset$, i.e. no bilateral Esscher martingale measure exists, which is again due to the behaviour of the cumulant generating function at its boundary. Consequently, the statement of Theorem \ref{thm-min-entropy-bilateral} does not hold true in this case. We require an extra condition on the parameters in order to ensure the existence of a minimal entropy martingale measure or of a $p$-optimal martingale measure. 
We emphasize that for tempered stable processes the function $\Phi$ defined in (\ref{def-phi}) is no longer available in closed form, which complicates the minimizing procedures of this section for concrete examples.
\end{remark}

\section{Existence of the minimal martingale measure in bilateral Gamma stock
models}\label{sec-mmm}

In this section, we deal with the existence of the \textit{minimal martingale measure} in bilateral Gamma stock models. The minimal martingale measure was introduced in \cite{FS} with the motivation of constructing optimal hedging strategies. 
Throughout this section, we fix a finite time horizon $T > 0$ and assume that $\lambda^+ > 2$. This allows us to define the finite value
\begin{align}\label{def-c}
c = c(\alpha^+,\alpha^-,\lambda^+,\lambda^-,r,q) = \frac{\Psi(1) - (r-q)}{\Psi(2) - 2\Psi(1)},
\end{align}
where $\Psi$ denotes the cumulant generating function defined in (\ref{cumulant}). Note that
\begin{align}\label{denom-pos}
\Psi(2) - 2\Psi(1) = \int_{\mathbb{R}} (e^{2x} - 1) F(dx) - 2 \int_{\mathbb{R}} (e^{x} - 1) F(dx) = \int_{\mathbb{R}} (e^{x} - 1)^2 F(dx) > 0,
\end{align}
whence (\ref{def-c}) is well-defined.
For technical reasons, we shall also assume that the filtration $(\mathcal{F}_t)_{t \geq 0}$ is generated by the bilateral Gamma process $X$.

In order to define the minimal martingale measure, we require the canonical decomposition of the discounted stock prices $\tilde{S}$ defined in (\ref{discounted-price}).

\begin{lemma}\label{lemma-canon-decomp}
The process $\tilde{S}$ is a special semimartingale with canonical decomposition
\begin{align}\label{canon-decomp}
\tilde{S}_t = S_0 + M_t + A_t, \quad t \in [0,T]
\end{align}
where $M$ and $A$ denote the processes
\begin{align}\label{M-part}
M_t &= \int_0^t \int_{\mathbb{R}} \tilde{S}_{s-} (e^x - 1) (\mu^X(ds,dx) - F(dx)ds), \quad t \in [0,T]
\\ A_t &= ( \Psi(1) - (r-q) ) \int_0^t \tilde{S}_{s} ds, \quad t \in [0,T].
\end{align}
The process $M$ is a square-integrable martingale and the finite variation part $A$ has the representation
\begin{align}\label{repr-FV-part}
A_t = \int_0^t \alpha_s d \langle M \rangle_s, \quad t \in [0,T]
\end{align}
where $\alpha$ denotes the predictable process
\begin{align}
\alpha_t = \frac{c}{\tilde{S}_{t-}}, \quad t \in [0,T].
\end{align}
\end{lemma}

\begin{remark}
In (\ref{M-part}), $\mu^X$ denotes the random measure associated to the jumps of $X$, and in (\ref{repr-FV-part}) the process $\langle M \rangle$ denotes the predictable quadratic variation of $M$, i.e., the unique predictable process with paths of finite variation such that $M^2 - \langle M \rangle$ is a local martingale, see, e.g., \cite[Sec. I.4a]{JS}.
\end{remark}

\begin{proof}
Note that we can express the discounted stock prices (\ref{discounted-price}) as
\begin{align*}
\tilde{S}_t = e^{Y_t}, \quad t \in [0,T]
\end{align*}
where $Y$ denotes the process
\begin{align*}
Y_t = X_t - (r-q)t, \quad t \in [0,T].
\end{align*}
Using It\^{o}'s formula \cite[Thm. I.4.57]{JS} we obtain
\begin{align*}
\tilde{S}_t &= S_0 + \int_0^t e^{Y_{s-}} dY_s + \sum_{s \leq t} \Big( e^{Y_s} - e^{Y_{s-}} - e^{Y_{s-}} \Delta Y_s \Big)
\\ &= S_0 + \int_0^t \int_{\mathbb{R}} x \tilde{S}_{s-} \mu^X(ds,dx) - (r-q) \int_0^t \tilde{S}_{s} ds 
\\ &\quad + \int_0^t \int_{\mathbb{R}} \tilde{S}_{s-} (e^x - 1 - x) \mu^X(ds,dx)
\\ &= S_0 + \int_0^t \int_{\mathbb{R}} \tilde{S}_{s-} (e^x - 1) (\mu^X(ds,dx) - F(dx)ds) 
\\ &\quad + \bigg( \int_{\mathbb{R}} (e^x - 1) F(dx) - (r-q) \bigg) \int_0^t \tilde{S}_{s} ds,
\end{align*}
and hence, $\tilde{S}$ is a special semimartingale with canonical decomposition (\ref{canon-decomp}). Recalling that $\lambda^+ > 2$, the process $M$ is a square-integrable martingale, because by (\ref{denom-pos}), (\ref{distribution-bil-time}) and (\ref{cumulant}) we have
\begin{align*}
&\mathbb{E} \bigg[ \int_0^T \int_{\mathbb{R}} | \tilde{S}_{s} (e^x - 1) |^2 F(dx) ds \bigg] = \bigg( \int_{\mathbb{R}} (e^x - 1)^2 F(dx) \bigg) \bigg( \int_0^T \mathbb{E} [ \tilde{S}_{s}^2 ] ds \bigg)
\\ &= S_0 ( \Psi(2) - 2 \Psi(1) ) \int_0^T e^{-(r-q)s} \mathbb{E}[e^{X_s}] ds
\\ &= S_0 ( \Psi(2) - 2 \Psi(1) ) \int_0^T e^{-(r-q)s} \bigg( \frac{\lambda^+}{\lambda^+ - 1} \bigg)^{\alpha^+ s} \bigg( \frac{\lambda^-}{\lambda^- + 1} \bigg)^{\alpha^- s} ds < \infty.
\end{align*}
Using \cite[Thm. II.1.33.a]{JS}, the predictable quadratic variation is given by
\begin{align*}
\langle M \rangle_t = \bigg( \int_{\mathbb{R}} (e^x - 1)^2 F(dx) \bigg) \int_0^t \tilde{S}_{s}^2 ds = ( \Psi(2) - 2 \Psi(1) ) \int_0^t \tilde{S}_{s}^2 ds, \quad t \in [0,T].
\end{align*}
Therefore, we deduce (\ref{repr-FV-part}).
\end{proof}

Lemma \ref{lemma-canon-decomp} shows that $\tilde{S}$ satisfies the so-called structure condition (SC) from \cite{Schweizer-MMM}. Therefore, according to \cite[Prop. 2]{Schweizer-MMM}, the stochastic exponential
\begin{align}\label{density-proc}
\hat{Z}_t = \mathcal{E} \bigg( - \int_0^{\bullet} \alpha_s dM_s \bigg)_t, \quad t \in [0,T]
\end{align}
is a \textit{martingale density} for $\tilde{S}$, that is $\hat{Z}$ and $\hat{Z} \tilde{S}$ are local $\mathbb{P}$-martingales and $\mathbb{P}(\hat{Z}_0 = 1) = 1$. The measure transformation
\begin{align}\label{MMM-trafo}
\frac{d \hat{\mathbb{P}}}{d \mathbb{P}} := \hat{Z}_T
\end{align}
defines a (possibly signed) measure $\hat{\mathbb{P}}$, the so-called \textit{minimal martingale measure} for $\tilde{S}$.

We are now interested in the question when $\hat{Z}$ is a \textit{strict martingale density}, that is, the martingale density $\hat{Z}$ is strictly positive. In this case, the martingale density $\hat{Z}$ is a strictly positive supermartingale, and hence the transformation (\ref{MMM-trafo}) defines (after normalizing, if necessary) a probability measure. If $\hat{Z}$ is even a strictly positive $\mathbb{P}$-martingale, then $\hat{\mathbb{P}}$ is a martingale measure for $\tilde{S}$.

\begin{theorem}
The following statements are equivalent.
\begin{enumerate}
\item $\hat{Z}$ is a strict martingale density for $\tilde{S}$.

\item $\hat{Z}$ is a strictly positive $\mathbb{P}$-martingale.

\item We have
\begin{align}\label{cond-Psi-1-2}
-1 \leq c \leq 0.
\end{align}

\item We have
\begin{align}\label{mmm-cond-1}
\bigg( \frac{\lambda^+}{\lambda^+ - 1} \bigg)^{\alpha^+} \bigg( \frac{\lambda^-}{\lambda^- + 1} \bigg)^{\alpha^-} &\leq e^{r-q}
\\ \label{mmm-cond-2} \text{and} \quad \bigg( \frac{\lambda^+ - 2}{\lambda^+ - 1} \bigg)^{\alpha^+} \bigg( \frac{\lambda^- + 2}{\lambda^- + 1} \bigg)^{\alpha^-} &\leq e^{-(r-q)}.
\end{align}

\end{enumerate}
If the previous conditions are satisfied, then under the minimal martingale measure $\hat{\mathbb{P}}$ we have
\begin{align}\label{X-under-MMM}
X \sim \Gamma((c+1) \alpha^+,\lambda^+;(c+1) \alpha^-,\lambda^-) * \Gamma(-c \alpha^+,\lambda^+ - 1;-c \alpha^-,\lambda^- + 1).
\end{align}
\end{theorem}

\begin{remark}
Relation (\ref{X-under-MMM}) means that under $\hat{\mathbb{P}}$ the driving process $X$ is the sum of two independent bilateral Gamma processes, that is, the characteristic function of $X_1$ under $\hat{\mathbb{P}}$ is given by
\begin{equation}\label{cf-MMM}
\begin{aligned}
\hat{\phi}(z) &= \bigg( \frac{\lambda^+}{\lambda^+ - iz} \bigg)^{(c+1)\alpha^+} \bigg( \frac{\lambda^-}{\lambda^- + iz} \bigg)^{(c+1)\alpha^-} 
\\ &\quad \times \bigg( \frac{\lambda^+ - 1}{\lambda^+ - 1 - iz} \bigg)^{-c \alpha^+} \bigg( \frac{\lambda^- + 1}{\lambda^- + 1 + iz} \bigg)^{-c \alpha^-}, \quad z \in \mathbb{R}.
\end{aligned}
\end{equation}
There are the following two boundary values.
\begin{itemize}
\item In the case $c = 0$ we have
\begin{align*}
X \sim \Gamma(\alpha^+,\lambda^+,\alpha^-,\lambda^-) \quad \text{under $\hat{\mathbb{P}}$,}
\end{align*}
i.e., the minimal martingale measure $\hat{\mathbb{P}}$ coincides with the physical measure $\mathbb{P}$. Indeed, the
definition (\ref{def-c}) of $c$ and Lemma \ref{lemma-martingale-measure} show that $\mathbb{P}$ already is a martingale measure for $\tilde{S}$.

\item In the case $c = -1$ we have
\begin{align*}
X \sim \Gamma(\alpha^+,\lambda^+-1,\alpha^-,\lambda^- + 1) \quad \text{under $\hat{\mathbb{P}}$,}
\end{align*}
i.e., the minimal martingale measure $\hat{\mathbb{P}}$ coincides with the Esscher transform $\mathbb{P}^1$, see Theorem \ref{thm-Esscher}. Indeed, the definition (\ref{def-c}) of $c$ shows that equation (\ref{Esscher-equation}) is satisfied with $\Theta = 1$.
\end{itemize}

\end{remark}

\begin{proof}
According to \cite[Prop. 2]{Schweizer-MMM} the process $\hat{Z}$ is a strict martingale density for $\tilde{S}$ if and only if we have almost surely
\begin{align*}
\alpha_t \Delta M_t < 1, \quad t \in [0,T].
\end{align*}
Taking into account Lemma \ref{lemma-canon-decomp}, this is satisfied if and only if we have almost surely
\begin{align*}
c (e^{\Delta X_t} - 1) < 1, \quad t \in [0,T]
\end{align*}
and this is the case if and only if we have (\ref{cond-Psi-1-2}). Noting (\ref{denom-pos}),
we observe that (\ref{cond-Psi-1-2}) is equivalent to the two conditions
\begin{align*}
\Psi(1) \leq r-q \quad \text{and} \quad \Psi(1) - \Psi(2) \leq - (r-q),
\end{align*}
and, in view of the cumulant generating function given by (\ref{cumulant}), these two conditions are fulfilled if and only if we have (\ref{mmm-cond-1}) and (\ref{mmm-cond-2}).

Now suppose that (\ref{cond-Psi-1-2}) is satisfied. We define the continuous functions
\begin{align*}
Y &: \mathbb{R} \rightarrow (0,\infty), \quad Y(x) := c + 1 - c e^x
\\ f &: (0,\infty) \rightarrow \mathbb{R}_+, \quad f(y) := y \ln y - (y-1).
\end{align*}
Then, by Lemma \ref{lemma-canon-decomp}, we can write the density process (\ref{density-proc}) as
\begin{align*}
\hat{Z}_t &= \mathcal{E} \bigg( - \int_0^{\bullet} \alpha_s dM_s \bigg)_t = \mathcal{E} \bigg( \int_0^{\bullet} \int_{\mathbb{R}} (c - c e^x) (\mu^X(ds,dx) - F(dx)ds) \bigg)_t
\\ &= \mathcal{E} \bigg( \int_0^{\bullet} \int_{\mathbb{R}} (Y(x) - 1) (\mu^X(ds,dx) - F(dx)ds) \bigg)_t, \quad t \in [0,T].
\end{align*}
Note that $Y$ is nondecreasing with $Y(0) = 1$. Therefore, we have
\begin{align*}
&\int_{\mathbb{R}} f(Y(x)) F(dx) \leq \int_{-\infty}^0 (1 - Y(x)) F(dx) + \int_0^{\infty} Y(x) \ln Y(x) F(dx).
\end{align*}
The first term is estimated as
\begin{align*}
\int_{-\infty}^0 (1 - Y(x)) F(dx) = -c \int_{-\infty}^0 (1 - e^x) F(dx) < \infty.
\end{align*}
Taking into account the estimates 
\begin{align*}
Y(x) &\leq 1 + e^x, \quad x \in \mathbb{R}
\\ \ln(1 + x) &\leq x, \quad x \geq 0
\end{align*}
we obtain, since $\lambda^+ > 2$ by assumption,
\begin{align*}
&\int_0^{\infty} Y(x) \ln Y(x) F(dx) \leq \int_0^{\infty} (e^x + 1) \ln (1 + c(1-e^x)) F(dx)
\\ &\leq -c \int_0^{\infty} (e^x + 1) (e^x - 1) F(dx) = -c \int_0^{\infty} (e^{2x} - 1) F(dx) < \infty.
\end{align*}
Consequently, we have
\begin{align*}
\int_{\mathbb{R}} f(Y(x)) F(dx) < \infty.
\end{align*}
By \cite[Cor. 6]{Esche-Schweizer} the process $\hat{Z}$ is a strictly positive $\mathbb{P}$-martingale.
Moreover, \cite[Prop. 7]{Esche-Schweizer} yields that $\beta = 0$ and $Y$ are the Girsanov parameters of $\hat{\mathbb{P}}$. Now \cite[Prop. 2]{Esche-Schweizer} gives us that $X$ is a L\'{e}vy process under $\hat{\mathbb{P}}$ whose characteristic triplet with respect to the truncation function $h \equiv 0$ is given by $(0,0,\hat{F})$, where, by taking into account (\ref{levy-measure}), the L\'{e}vy measure is given by
\begin{equation}\label{levy-measure-MMM}
\begin{aligned}
\hat{F}(dx) &= Y(x) F(dx) = ( c + 1 -c e^x ) F(dx) 
\\ &= \bigg[ \bigg( \frac{(c+1) \alpha^+}{x} e^{-\lambda^+ x} + \frac{-c \alpha^+}{x} e^{-(\lambda^+ - 1)x} \bigg) \mathbbm{1}_{(0,\infty)}(x)
\\ &\quad + \bigg( \frac{(c+1)\alpha^-}{|x|}  e^{-\lambda^- |x|} + \frac{-c \alpha^-}{|x|} e^{-(\lambda^- + 1)|x|} \bigg) \mathbbm{1}_{(-\infty,0)}(x) \bigg] dx.
\end{aligned}
\end{equation}
Thus, the distribution of $X$ under $\hat{\mathbb{P}}$ is given by (\ref{X-under-MMM}).
\end{proof}

We proceed with an application to quadratic hedging. Here we assume that $\lambda^+ > 3$, which ensures that $\hat{\Psi}(1)$ and $\hat{\Psi}(2)$ appearing in (\ref{Delta-strategy}) below are well-defined.
Let $\Phi : \mathbb{R} \rightarrow \mathbb{R}$ be a payoff function. The prices at time $t \in [0,T]$ are given by $\pi(t,S_t)$, where
\begin{align*}
\pi(t,S) = \mathbb{E}_{\hat{\mathbb{P}}} [ \Phi(S_T) \, | \, S_t = S ].
\end{align*}
For example, the price at time $t$ of a European call option with strike price $K$ and maturity $T$ is given by
\begin{align*}
\pi(t,S) &= - \frac{e^{-r(T-t)} K}{2 \pi} \int_{i \nu - \infty}^{i \nu + \infty} \bigg( \frac{K}{S} \bigg)^{iz} \bigg( \frac{\lambda^+}{\lambda^+ + iz} \bigg)^{(c+1)\alpha^+ (T-t)} 
\\ &\qquad\qquad\qquad\qquad \times \bigg( \frac{\lambda^-}{\lambda^- - iz} \bigg)^{(c+1)\alpha^- (T-t)} \bigg( \frac{\lambda^+ - 1}{\lambda^+ - 1 + iz} \bigg)^{-c \alpha^+ (T-t)} 
\\ &\qquad\qquad\qquad\qquad \times \bigg( \frac{\lambda^- + 1}{\lambda^- + 1 - iz} \bigg)^{-c \alpha^- (T-t)} \frac{dz}{z(z-i)},
\end{align*}
where $\nu \in (1,\lambda^+ - 1)$ is arbitrary. This follows from the form of the characteristic function (\ref{cf-MMM}) and \cite[Thm. 3.2]{Lewis}.

Following \cite[Sec. 10.4]{Cont-Tankov}, we choose our trading strategy $\xi$ as
\begin{align*}
\xi_t = \Delta(t,S_t), \quad t \in [0,T]
\end{align*}
with $\Delta$ given by
\begin{equation}\label{Delta-strategy}
\begin{aligned}
\Delta(t,S) &= \frac{\int_{\mathbb{R}} (e^x - 1) (\pi(t,S e^x) - \pi(t,S)) \hat{F}(dx)}{S \int_{\mathbb{R}} (e^x - 1)^2 \hat{F}(dx)}
\\ &= \frac{\int_{\mathbb{R}} (e^x - 1) (\pi(t,S e^x) - \pi(t,S)) \hat{F}(dx)}{S (\hat{\Psi}(2) - 2 \hat{\Psi}(1))},
\end{aligned}
\end{equation}
where $\hat{F}$ denotes the L\'{e}vy measure derived in (\ref{levy-measure-MMM}) and $\hat{\Psi}$ denotes the cumulant generating function of $X_1$ under $\hat{\mathbb{P}}$ given by
\begin{align*}
\hat{\Psi}(z) &= (c+1)\alpha^+ \ln \bigg( \frac{\lambda^+}{\lambda^+ - z} \bigg) + (c+1) \alpha^- \ln \bigg( \frac{\lambda^-}{\lambda^- + z} \bigg)
\\ &\quad -c \alpha^+ \ln \bigg( \frac{\lambda^+ - 1}{\lambda^+ - 1 - z} \bigg) - c \alpha^- \ln \bigg( \frac{\lambda^- + 1}{\lambda^- + 1 + z} \bigg), \quad z \in (-\lambda^-,\lambda^+ - 1).
\end{align*}
Note that the strategy $\xi$ is, in general, not self-financing. It minimizes the hedging error
\begin{align*}
\mathbb{E}_{\hat{\mathbb{P}}} \Bigg[ \bigg| \Phi(S_T) - \bigg( \mathbb{E}_{\hat{\mathbb{P}}}[\Phi(S_T)] + \int_0^T \xi_t dS_t \bigg) \bigg|^2 \Bigg]
\end{align*}
with respect to the $L^2(\mathbb{\hat{P}})$-distance. Therefore, we have
\begin{align}\label{FS-decomp}
\Phi(S_T) = \mathbb{E}_{\hat{\mathbb{P}}} [\Phi(S_T)] + \int_0^T \xi_t dS_t + N_T,
\end{align}
where $N$ is a local martingale such that $\int_0^T \xi_t dS_t$ and $N_T$ are orthogonal in $L^2(\hat{\mathbb{P}})$. 
The decomposition (\ref{FS-decomp}) is the so-called \textit{F\"ollmer-Schweizer decomposition}. The orthogonal component $N$ is the intrinsic risk which cannot be hedged. Since the measure change from $\hat{\mathbb{P}}$ to $\mathbb{P}$ preserves orthogonality, $\int_0^T \xi_t dS_t$ and $N_T$ are also orthogonal in $L^2(\mathbb{P})$. Therefore, the strategy $\xi$ also minimizes the quadratic hedging error
\begin{align*}
\mathbb{E}_{\mathbb{P}} \Bigg[ \bigg| \Phi(S_T) - \bigg( \mathbb{E}_{\hat{\mathbb{P}}}[\Phi(S_T)] + \int_0^T \xi_t dS_t \bigg) \bigg|^2 \Bigg]
\end{align*}
with respect to the physical probability measure $\mathbb{P}$.

\begin{remark}
Analogous results of this section are also valid for general tempered stable processes. The concrete conditions (\ref{mmm-cond-1}), (\ref{mmm-cond-2}) on the parameters differ considerably in the tempered stable case, due to the different form of the cumulant generating function $\Psi$.
\end{remark}

\section{A numerical illustration}\label{sec-numerics}

\begin{figure}[t]
   \centering
   \includegraphics[height=40ex,width=50ex]{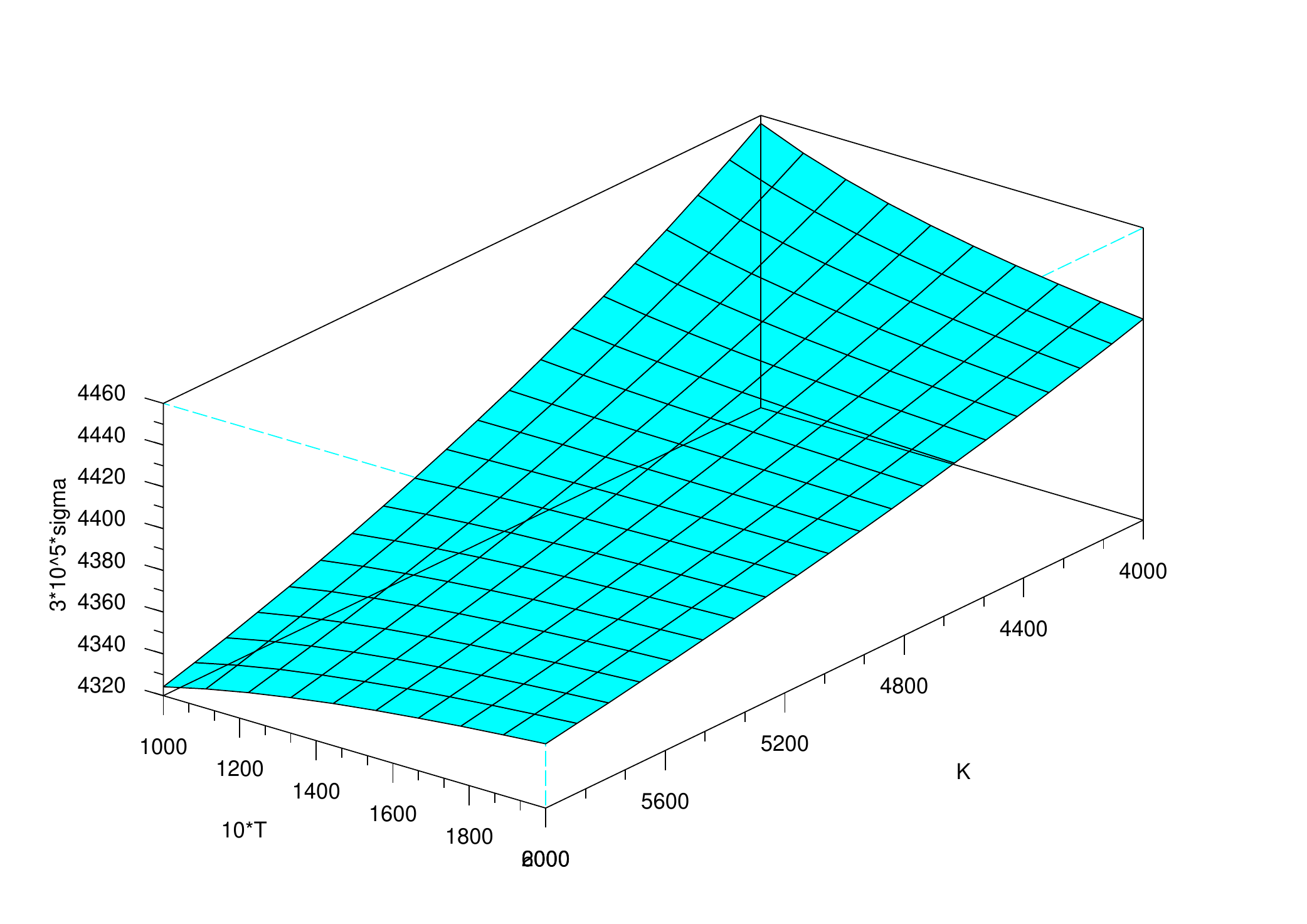}
   \caption{Implied volatility surface.}
   \label{fig-impl-vol}
\end{figure}

We conclude this article with a numerical illustration of the
preceding results. For a bilateral Gamma stock model of the type
(\ref{exp-Levy-model}) we have estimated the parameters of the
bilateral Gamma process $X$ as
\begin{align*}
(\alpha^+,\lambda^+;\alpha^-,\lambda^-) = (1.55,133.96;0.94,88.92)
\end{align*}
from observations of the German stock index DAX, see \cite[Sec.
9]{Kuechler-Tappe}. For our upcoming calculations, we take the
initial stock price $S_0 = 5000$ and, for simplicity, we put $r = q
= 0$.

We start with the computation of the Esscher transform from Section
\ref{sec-Esscher}. Note that condition
(\ref{cond-for-Esscher-lambda}) is fulfilled. Solving equation
(\ref{Esscher-equation}), we obtain the Esscher transform
$\mathbb{P}^{\Theta}$ with $\Theta = -5.28$. Under the measure
$\mathbb{P}^{\Theta}$, the driving process $X$ is bilateral Gamma
with parameters
\begin{align*}
(\alpha^+,\lambda^+ - \Theta;\alpha^-,\lambda^- + \Theta) =
(1.55,139.24;0.94,83.64).
\end{align*}
The value of the relative entropy, provided by
(\ref{insert-Esscher}), is given by
\begin{align}\label{entropy-1}
\mathbb{H}(\mathbb{P}^{\Theta} \,|\, \mathbb{P}) = 0.00294113.
\end{align}
We proceed with the minimal entropy martingale measure from Section
\ref{sec-minimal-entropy}. We have
\begin{align}\label{rhs}
\alpha^+ \ln \left( \frac{\lambda^+}{\lambda^+ - 1} \right) +
\alpha^- \ln \left( \frac{\lambda^-}{\lambda^- + 1} \right) =
0.0000221,
\end{align}
whence condition (\ref{entropy-inequality}) is satisfied. Solving
(\ref{eqn-minimize-entropy}), we obtain the minimal entropy
martingale measure $\mathbb{P}_{\vartheta}$ with $\vartheta =
-5.30$. Using (\ref{value-entropy}), the value of the minimal
entropy is given by
\begin{align}\label{entropy-2}
\mathbb{H}(\mathbb{P}_{\vartheta} \,|\, \mathbb{P}) = 0.00294091.
\end{align}
Finally, we turn to the computation of the bilateral Esscher
transform from Section \ref{sec-bilateral-Esscher}. Minimizing the
function $f$ given by (\ref{def-f}), we obtain the bilateral Esscher
transform $\mathbb{P}^{(\theta,\Phi(\theta))}$ with $\theta =
-5.34$. We observe that the Esscher parameter $\Theta$ and the
bilateral Esscher parameter $\theta$ are quite close to each other,
see Remark \ref{remark-Esscher}. Under the measure
$\mathbb{P}^{(\theta,\Phi(\theta))}$, the driving process $X$ is
bilateral Gamma with parameters
\begin{align*}
(\alpha^+,\lambda^+ - \theta;\alpha^-,\lambda^- - \Phi(\theta)) =
(1.55,139.30;0.94,83.68).
\end{align*}
Computing the relative entropy according to (\ref{compute-entropy})
yields
\begin{align}\label{entropy-3}
\mathbb{H}(\mathbb{P}^{(\theta,\Phi(\theta))} \,|\, \mathbb{P}) =
0.00294107.
\end{align}
The relative entropies computed in (\ref{entropy-1}),
(\ref{entropy-2}), (\ref{entropy-3}) show that both, the Esscher
transform $\mathbb{P}^{\Theta}$ and the bilateral Esscher transform
$\mathbb{P}^{(\theta,\Phi(\theta))}$ are quite close to the minimal
entropy martingale measure $\mathbb{P}_{\vartheta}$.

For $p > 1$ we can compute the $p$-optimal martingale measure within the class of bilateral Gamma processes by minimizing the function $f_p$ given by (\ref{def-f-p}). In particular, we obtain the variance-optimal martingale measure as the bilateral Esscher transform $\mathbb{P}^{(\theta_2,\Phi(\theta_2))}$ with $\theta_2 = -5.68$. Furthermore, we observe that $\theta_p \rightarrow \theta$ for $p \downarrow 1$, a behaviour, which we have discussed at the end of Section \ref{sec-bilateral-Esscher}.

The minimal martingale measure from Section \ref{sec-mmm} does not exist for the present parameters, because the constant $c$ defined in (\ref{def-c}) is positive. However, introducing a small interest rate, e.g. we could take $r = 0.0012$, then condition (\ref{cond-Psi-1-2}) is fulfilled, and hence the minimal martingale measure $\hat{\mathbb{P}}$ exists.

Going back to the minimal entropy martingale measure $\mathbb{P}^{(\theta,\Phi(\theta))}$,
we can numerically compute the prices of European call options by
using formula (\ref{option-pricing}). Figure \ref{fig-impl-vol}
shows the implied volatility surface. We observe the following
properties:

\begin{itemize}
\item The dependence of the implied volatility with respect to the
strike price $K$ is decreasing, we have a so-called "skew".

\item The skew flattens out for large times of maturity $T$.
\end{itemize}

Hence, the implied volatility surface in Figure \ref{fig-impl-vol}
has the typical features, which one observes in practice.

For any positive integer $n \in \mathbb{N}$ we can express $Y \sim \Gamma(n \alpha^+,\lambda^+;n \alpha^-, \lambda^-)$ as $Y = X_1 + \ldots + X_n$, where the $X_i$ are i.i.d. with $X_i \sim \Gamma(\alpha^+,\lambda^+,\alpha^-,\lambda^-)$. Therefore, if the parameters are fixed, a random variable $X \sim \Gamma(\alpha^+ t,\lambda^+,\alpha^- t,\lambda^-)$ has the tendency to be approximately normally distributed for increasing $t$, and therefore the implied volatility surface becomes flatter with increasing $t$. This observation also explains why the skew in Figure \ref{fig-impl-vol} flattens out for large times of maturity $T$.

\subsection*{Acknowledgement}

The second author gratefully acknowledges the support from WWTF (Vienna Science and Technology Fund).

We are grateful to Damir Filipovi\'c, Friedrich Hubalek, Katja Krol, Michael Kupper and Antonis Papapantoleon for their helpful remarks and discussions. We also thank two anonymous referees for their helpful comments and suggestions.

\end{document}